\documentclass[conference,a4paper]{IEEEtran}
\IEEEoverridecommandlockouts
\usepackage{url}
\usepackage[english]{babel}
\usepackage{amsmath,amsthm}
\usepackage{amssymb,mathrsfs}
\usepackage{amsfonts}
\usepackage[final]{graphicx}
\usepackage{epstopdf}
\usepackage{tikz}
\usepackage{pgfplots}
\usepackage{subfig}
\usepackage{cprotect}

\usepackage{tikz}
\usetikzlibrary{arrows,positioning,fit,backgrounds}
\usetikzlibrary{calc}

\usetikzlibrary{arrows,positioning,fit,backgrounds,shapes}
\usetikzlibrary{calc}
\newtheorem{thm}{Theorem}

\newtheorem{lem}[thm]{Lemma}

\newtheorem{defn}{Definition}
\theoremstyle{remark}
\newtheorem{rem}{Remark}

\newcommand{\mb}{\mathbf}
\DeclareMathOperator*{\Cov}{Cov}
\DeclareMathOperator*{\Tr}{Tr}

\newcommand{\uh}{\hat{U}}
\newcommand{\bsigma}{\mb{\Sigma}}

\usepackage{xcolor}

\begin{document}
\title{
Brascamp-Lieb Inequality and Its Reverse: \\An Information Theoretic View \thanks{This work was supported in part by  NSF Grants CCF-1528132,
CCF-0939370 (Center for Science of Information),
CCF-1319299,
CCF-1319304,
CCF-1350595 and AFOSR FA9550-15-1-0180.}
}%

\author{
\IEEEauthorblockN{Jingbo Liu\IEEEauthorrefmark{1}, Thomas~A.~Courtade\IEEEauthorrefmark{2}, Paul Cuff\IEEEauthorrefmark{1} and Sergio Verd\'{u}\IEEEauthorrefmark{1}}\\
\IEEEauthorblockA{ \IEEEauthorrefmark{1}Department of Electrical Engineering, Princeton University\\
 \IEEEauthorrefmark{2}Department of Electrical Engineering and Computer Sciences, University of California, Berkeley \\
 Email:  \{{jingbo,cuff,verdu}\}@princeton.edu, courtade@eecs.berkeley.edu
           \vspace{-2ex}   }
}

%

\maketitle

\begin{abstract}
We generalize a result by Carlen and Cordero-Erausquin on the equivalence between the Brascamp-Lieb inequality and the subadditivity of relative entropy
by allowing for random transformations (a broadcast channel).
This leads to a unified perspective on several functional inequalities that have been gaining popularity in the context of proving  impossibility results.
We demonstrate that the information theoretic dual of the Brascamp-Lieb inequality is a convenient setting for proving properties such as data processing, tensorization, convexity and Gaussian optimality.  Consequences of the latter include an extension of the Brascamp-Lieb inequality allowing for Gaussian random transformations, the determination of the multivariate Wyner common information for Gaussian sources, and a multivariate version of Nelson's hypercontractivity theorem.
Finally we present an information theoretic characterization of a reverse Brascamp-Lieb inequality involving a random transformation (a multiple access channel).
\end{abstract}

\section{Introduction}\label{sec1}
The Brascamp-Lieb (BL) inequality \cite{brascamp1976best}\cite{lieb1990gaussian}\cite{bennett2008brascamp} in functional analysis
concerns the optimality of Gaussian functions in a certain class of integral inequalities.
To be concrete, consider an inequality of the following general form,
which we shall call a \emph{Brascamp-Lieb like} (BLL) inequality:
\begin{align}
\mathbb{E}\left[\prod_{j=1}^m f_j(B_j(X))
\right]
\le D\prod_{j=1}^m \|f_j\|_{\frac{1}{c_j}},
\label{e_carlen}
\end{align}
where the expectation is with respect to $X\sim Q$, for each $j\in\{1,\dots,m\}$, $c_j\in(0,\infty)$, $B_j\colon \mathcal{X}\to\mathcal{Y}_j$ is measurable, $f_j\colon \mathcal{Y}_j\to \mathbb{R}$ is nonnegative measurable, and
$$\|f_j\|_{\frac{1}{c_j}}:=\left(\int f_j^{1/c_j}{\rm d}R_{Y_j}\right)^{c_j}$$ for some $R_{Y_j}$.
Conventionally, the BL inequality is the Gaussian case of \eqref{e_carlen} where $Q$ is Gaussian (or the Lebesgue measure, with the expectations replaced by integrals)
and $(B_j)_{j=1}^m$ are linear projections.
In this setting, Brascamp and Lieb \cite{brascamp1976best} showed that \eqref{e_carlen} holds if and only if it holds for all centered Gaussian functions $(f_j)_{j=1}^m$.
Generalizing a result in \cite{brascamp1976best}, Lieb \cite{lieb1990gaussian} extended the validity of the result to arbitrary surjective linear maps $(B_j)$.
Lieb's proof
used a rotational invariance property of Gaussian random variables.
Given the fundamental nature of Lieb's result and its far-reaching consequences, alternative proof methods have attracted wide interest; see \cite[Remark~1.10]{bennett2008brascamp} for the history and references.

Motivated by the quest for an alternative proof of Lieb's result using the superadditivity of Fisher information,
Carlen and Cordero-Erausquin proved a duality between the BLL inequality in \eqref{e_carlen} and a super-additivity of relative entropy.
We remind the reader that the BLL inequality is more general than the setting initially considered by Brascamp and Lieb \cite{brascamp1976best} since the random variables need not be Gaussian and $(B_j)$ need not be linear.


In this paper, we extend the duality result of Carlen and Cordero-Erausquin, by allowing $(B_j)$ to be non-deterministic random transformations (that is, there is a broadcast channel). This is motivated by two considerations.
First, random transformations are natural and essential for many information theoretic applications; see for example \cite{ISIT_lccv_smooth2016}. Second, the result subsumes the equivalent formulation of the strong data processing inequality (in addition to that of hypercontractivity and Loomis-Whitney inequality/Shearer's lemma, which is already contained in Carlen and Cordero-Erausquin's result).
These inequalities have recently attracted significant attention in information theory
\cite{courtade2013outer}\cite{pw_2015}\cite{liu2015key}\cite{xu15},
theoretical computer science
\cite{ganor2014exponential}\cite{braverman2015communication}
and statistics
\cite{mossel2010noise}\cite{duchi13}.
Previous proofs of their equivalent formulations have been discovered independently and sometimes rely on the finiteness of the alphabet.
In contrast, the present approach is based on the nonnegativity of relative entropy (which corresponds to the Donsker-Varadhan formula used by Carlen and Cordero-Erausquin)
and holds for general alphabets.

In the same vein as Brascamp and Lieb's original result that Gaussian kernels have Gaussian maximizers, we establish Gaussian optimality in several information theoretic optimization problems related to the dual form of the BL inequality.
Roughly speaking, our approach is based on the fact that two independent random variables are both Gaussian if their sum is independent of their difference, a fact which was also used in establishing Gaussian extremality in information theory by Geng-Nair \cite{geng2014yanlin} \cite{nairextremal}.
It is worth noting that similar techniques have appeared previously in the literature on the Brascamp-Lieb inequality:  Lieb used a rotational invariance argument in \cite{lieb1990gaussian}, and it was observed that convolution preserves the extremizers of Brascamp-Lieb inequality \cite[Lemma~2]{barthe1998optimal}.
However, as keenly noted in \cite{geng2014yanlin}, working with the information theoretic counterparts offers certain advantages,
partly because of the similarity between
the proof techniques with certain converses in data transmission.
Implications of Gaussian optimality are discussed in Section~\ref{sec_gaussianOPT}.

Finally, we provide an information theoretic formulation of Barthe's reverse Brascamp-Lieb inequality (RBL) \cite{barthe1998reverse}.
In fact, we shall consider a generalization of RBL involving a multiple access channel (MAC) - pairing nicely with the broadcast channel in  the forward inequality.
From this formulation it is seen that the mysterious ``sup'' operation in RBL disappears when the MAC is a bijective mapping, which is the special case of reverse hypercontractivity considered by Kamath \cite{kamath_reverse}. Moreover, the strong data processing inequality is also a special case of generalized RBL when the MAC is a point to point channel.
Another direction of generalizing the reverse hypercontractivity, where the stochastic map is still a bijective mapping but the coefficients $(c_j)$ are allowed to be negative, has been recently considered by Beigi and Nair (see \cite{BN2015}). To our knowledge, their result does not imply ours, or vice versa.

Omitted proofs are given in \cite{lccv2015}.

\section{A General Duality Result}\label{sec_dual}

Given two probability measures
$P\ll Q$ on $\mathcal{X}$, define the \emph{relative information} as the log Radon-Nikodym derivative
\begin{align}
\imath_{P\|Q}(x):=\log\frac{{\rm d}P}{{\rm d}Q}(x)
\end{align}
The relative entropy is defined as
\begin{align}
D(P\|Q):=\mathbb{E}\left[\imath_{P\|Q}(X)\right]
\end{align}
where $X\sim P$, if $P\ll Q$, and infinity otherwise.

\begin{thm}\label{thm_1}
Fix $Q_X$, positive integer $m$, $d\in\mathbb{R}$, and $Q_{Y_j|X}$, measure $R_{Y_j}$ on $\mathcal{Y}_j$, $c_j\in(0,\infty)$ for each $j\in\{1,\dots,m\}$. Let $(X,Y_j)\sim Q_XQ_{Y_j|X}$.
Then the following statements are equivalent:
\begin{enumerate}
  \item For any non-negative measurable functions $f_j\colon\mathcal{Y}_j\to \mathbb{R}$,
      \begin{align}
      \mathbb{E}\left[\exp\left(\sum_{j=1}^m\mathbb{E}[\log f_j(Y_j)|X]-d\right)\right]
      &\le \prod_{j=1}^m\|f_j\|_{\frac{1}{c_j}}.\label{e_func}
      \end{align}
  \item For $P_X\ll Q_X$ and $P_X\to Q_{Y_j|X}\to P_{Y_j}$,
      \begin{align}
      D(P_X\|Q_X)+d\ge\sum_{j=1}^m c_jD(P_{Y_j}\|R_{Y_j}).
        \label{e_info}
      \end{align}
\end{enumerate}
\end{thm}
The special case where $Q_{Y_j|X}$, $j\in\{1,\dots,m\}$ are deterministic is established in \cite[Theorem~2.1]{carlen2009subadditivity}. We refer to \eqref{e_func} as a \emph{generalized Brascamp-Lieb like} (GBLL) inequality.

\begin{proof}[Proof Sketch]
The key idea of the proof is to define certain auxiliary distributions. Later we will reveal a nice symmetry with the auxiliary distributions used in the proof of the reverse inequality.

1)$\Rightarrow$2) Define an auxiliary measure $S_X$ via
  \begin{align}
  \imath_{S_X\|Q_X}(x):=-d_0+\sum_{j=1}^m c_j\mathbb{E}
  [\imath_{P_{Y_j}\|R_{Y_j}}(Y_j)|X=x]
  \nonumber
  \end{align}
  where $d_0$ is a normalization constant, and
    $
    f_j\leftarrow \left(\frac{{\rm d}P_{Y_j}}{{\rm d}R_{Y_j}}\right)^{c_j}
    $.
    Then 2) follows from 1) and the nonnegativity of  $D(P_X\|S_X)$.

2)$\Rightarrow$1) Define $P_X$ and $S_{Y_j}$ through
      \begin{align}
        \imath_{P_X\|Q_X}(x)=-d&-d_0+\mathbb{E}
        \left[\left.\sum_{j=1}^m\log f_j(Y_j)\right|{ X=x}\right]
        \label{e12}
        \\
  \imath_{S_{Y_j}\|R_{Y_j}}({ y_j})&:=\frac{1}{c_j}\log f_j({ y_j})-d_j,
  \end{align}
  where $d_j$'s are normalization constants. Then 2) follows from 1) and the nonnegativity of $D(P_{Y_j}\|S_{Y_j})$.
\end{proof}

\section{Notable Special Cases of Theorem~\ref{thm_1}}\label{sec_special}
Not only does Theorem~\ref{thm_1} admit a very simple proof but it
unifies the equivalent formulations of several functional inequalities and information theoretic inequalities,
and the approach applies to general alphabets, in contrast to some previous methods requiring finite alphabets (cf. \cite{ahlswede1976spreading,nair}).

\subsection{Variational Formula for R\'{e}nyi Divergence}
\label{sec_lpcb}
Suppose $R$, $Q$ and $T$ are probability measures on $(\mathcal{X},\mathscr{F})$, $R,Q\ll T$, $\alpha\in(0,1)\cup(1,\infty)$. Define $D_{\alpha}(Q\|R)$ as
\begin{align}
\frac{1}{\alpha-1}\log\left(\mathbb{E}\left[\exp\left(
\alpha\imath_{Q\|T}(\bar{X})+
(1-\alpha)\imath_{R\|T}(\bar{X})\right)\right]\right)\label{e18}
\end{align}
where $\bar{X}\sim T$.
\cite{atarRobust} showed that \eqref{e18} equals the supremum of
\begin{align}
&\frac{\alpha}{\alpha-1}\log\mathbb{E}[ \exp((\alpha-1)g(\hat{X}))]-
\log\mathbb{E}[\exp(\alpha g(X))]
\label{e_62}
\end{align}
over bounded nonnegative measurable $g$
such that \eqref{e_62} is well-defined,
where $X\sim {R}$ and $\hat{X}\sim{Q}$.
For $\alpha\in(1,\infty)$,
by setting $\exp(g(\cdot))$ as an indicator function of a set, one recovers the logarithmic probability comparison bound (LPCB) \cite{atar2014information}, which is useful in the error exponent analysis.

Now, a simple proof of \eqref{e_62} for $\alpha\in(1,\infty)$ can be obtained from Theorem~\ref{thm_1} by setting $m\leftarrow 1$, $Y_1=X$, $c\leftarrow \frac{\alpha-1}{\alpha}$ and $d=\frac{\alpha-1}{\alpha}D(P\|{R})$.
Note that \eqref{e_info} is then reduced to
\begin{align}
D(P\|Q)
+\frac{\alpha-1}{\alpha}D_{\alpha}({Q}\|{R})
\ge\frac{\alpha-1}{\alpha}D(P\|{R})
\label{e_64}
\end{align}
which is well-known and can be easily shown using the nonnegativity of relative entropy.
Meanwhile, setting $f\leftarrow \exp((\alpha-1)g)$ in
\eqref{e_func}, we see
\eqref{e_62}
is less than or equal to $D_{\alpha}(Q\|R)$. The equality is achieved when
\begin{align}
\frac{{\rm d}{Q}}{{\rm d}{R}}(x)=\frac{\exp(\alpha g(x))}{\mathbb{E}[\exp(\alpha g(X))]}.
\end{align}

\subsection{Strong Data Processing Constant}
A strong data processing inequality (SDPI) is an inequality of the form
\begin{align}
D(P_X\|Q_X)\ge c D(P_Y\|Q_Y),\quad \textrm{for all $P_X\ll Q_X$}
\label{e_sdp}
\end{align}
where $P_X\to Q_{Y|X}\to P_Y$, and we have fixed $Q_{XY}=Q_XQ_{Y|X}$ \cite{ahlswede1976spreading}\cite{csiszar2011information}\cite{anan_13}. The conventional data processing inequality corresponds to
$c=1$, so SDPI's generally specify $c>1$.
The study of the largest constant $c$ for \eqref{e_sdp} to hold can be traced back to Ahlswede and G\'{a}cs \cite{ahlswede1976spreading}, who
showed its equivalence to
\begin{align}
\mathbb{E}[\exp(\mathbb{E}[\log f(Y)|X])]\le \|f\|_{\frac{1}{c}},\quad\textrm{for all $f\ge0$}.
\label{e_func_sdpi}
\end{align}
The proof in \cite[Theorem~5]{ahlswede1976spreading}, which is based on a connection between SDPI and hypercontractivity, relies heavily on the finiteness of the alphabet, and is quite technical even in that case.
From Theorem~\ref{thm_1}, however,
it is straightforward to check that such equivalence holds for general alphabets.

\subsection{Hypercontractivity}
The BLL inequality also encompasses
\begin{align}
\mathbb{E}\left[f_1(Y_1)f_2(Y_2)\right]
&\le \|f_1\|_{p_1}
\|f_2\|_{p_2}
\label{e_sdp3}
\end{align}
where $p_1,p_2\in[0,\infty)$.
Using the method of types/typicality, it is shown in \cite{nair} that  \eqref{e_sdp3} is equivalent to
\begin{align}
D(P_{Y_1Y_2}\|Q_{Y_1Y_2})
\ge 
\tfrac{1}{p_1}D(P_{Y_1}\|Q_{Y_1})
+\tfrac{1}{p_2}D(P_{Y_2}\|Q_{Y_2})
\label{e25}
\end{align}
for all $P_{Y_1Y_2}\ll Q_{Y_1Y_2}$ in the case of finite alphabets.
The proof of Theorem~\ref{thm_1} based on nonnegativity of relative entropy establishes this equivalence for general alphabets,
and in particular allows one to prove Nelson's inequality for Gaussian hypercontractivity from \eqref{e25}; see the end of Section~\ref{sec_gaussianOPT}.

\subsection{Loomis-Whitney Inequality and Shearer's Lemma}
The combinatorial Loomis-Whitney inequality \cite[Theorem~2]{loomis1949} can be recovered from the following integral inequality: let $\mu$ be the counting measure on $\mathcal{A}^m$, then
\begin{align}
\int_{\mathcal{A}^m}\prod_{j=1}^m f_{j}(\pi_j(x))
{\rm d}\mu(x)
\le \prod_{j=1}^m \|f_{j}\|_{m-1}
\label{e_ilw}
\end{align}
for all nonnegative $f_j$'s, where
where $\pi_j$ is the projection operator deleting the $j$-th coordinate
and the norm on the right is with respect to the counting measure on $\mathcal{A}^{m-1}$.
This is an extension of the BLL inequality to counting measures\footnote{Theorem~\ref{thm_1} allows obvious extensions to nonnegative $\sigma$-finite measures.}, and by Theorem~\ref{thm_1} is equivalent to the entropy inequality known as Shearer's Lemma \cite{shearer}
\begin{align}
H(X_1,\dots,X_m)
\le \sum_{j=1}^m\frac{1}{m-1}
H(X_1^{j-1},X_{j+1}^{m}).
\end{align}

\section{Applications of the Information Theoretic Formulation}\label{sec_gaussian}
\subsection{Data Processing, Tensorization and Convexity}
The information theoretic formulation
in Theorem~\ref{thm_1} leads to simple proofs of basic and important properties of BLL inequalities.
Assuming $R_{Y_j}=Q_{Y_j}$ for simplicity, one has \cite{lccv2015}:
\begin{itemize}
  \item \emph{Data processing}: if $(Q_X,(Q_{Y_j|X}),d,c^m)$ is such that \eqref{e_info} holds, then by data processing for the relative entropy, $(Q_X,(Q_{Z_j|X}),d,c^m)$ also holds for any $(Q_{Z_j|Y_j})$, where $Q_{Y_j|X}\to Q_{Z_j|Y_j}\to Q_{Z_j|X}$. A similar property holds for processing the input.
  \item \emph{Tensorization}: if $(Q_X^{i},(Q_{Y_j|X}^{i}),d^{i},c^m)$, $i=1,2$ satisfies \eqref{e_info}, then $(Q_X^1\times Q_X^2,(Q_{Y_j|X}^1\times Q_{Y_j|X}^2),d^1+d^2,c^m)$ satisfies \eqref{e_info}. The proof is similar to standard converse proofs in information theory.
  \item \emph{Convexity}: if $(Q_X^{i},(Q_{Y_j|X}^{i}),d^{i},({c^i}_j))$ $i=1,2$ satisfies \eqref{e_info}, then for any $\theta\in(0,1)$, $(Q_X^{\theta},(Q_{Y_j|X}^{\theta}),d^{\theta},({c^{\theta}}_j))$ also satisfies \eqref{e_info} where
      \begin{align}
        d^{\theta}&:=(1-\theta)d^0+\theta d^1,
        \\
        c_j^{\theta}&:=(1-\theta)c_j^0+\theta c_j^1,\quad\forall j\in\{1,\dots,m\}.
        \end{align}
  This is equivalent to the Riesz-Thorin interpolation theorem in functional analysis in the case of non-negative kernels. The proof in the information theoretic setting is much simpler because the $c_j$'s only affect the right side of \eqref{e_info} as linear coefficients.
\end{itemize}
\subsection{Gaussian Optimality}\label{sec_gaussianOPT}
A less direct application is found in establishing Gaussian optimality for several information theoretic inequalities related to the BL inequalities.
Toward this end, assume for the remainder of the section that $\mathcal{X},\mathcal{Y}_1,\dots,\mathcal{Y}_m$
are Euclidean spaces of dimensions $n,n_1,\dots,n_m$, respectively, and
$Q_{\bf X}$ and $(Q_{{\bf Y}_j|\bf X})$ are Gaussian. To be precise about the notions of Gaussian optimality, we adopt terminology from \cite{bennett2008brascamp}:
\begin{itemize}
  \item \emph{Extremisability}: sup/inf is finitely attained.
  \item \emph{Gaussian extremisability}: sup/inf is finitely attained by Gaussian functions/distributions.
  \item \emph{Gaussian exhaustibility}: the value of sup/inf does not change when the arguments are restricted to Gaussian functions/distributions.
\end{itemize}
Most of the time, we prove Gaussian exhaustibility in general, while the more restrictive property of Gaussian extremisability is shown imposing non-degeneracy assumptions.

Fix $\mb{M}\succeq 0$,  positive constants $c_j$, and Gaussian random transformations $Q_{\mb{Y}_j|\mb{X}}$ for $j\in\{1,\dots,m\}$. Define
\begin{align}
F(P_{\mb{X}U}):= -h(\mb{X}|U) +\sum_{j=1}^m c_j h(\mb{Y}_j|U)+c_0 \Tr [\mb{M} \mb{\Sigma}_{\mb{X}|U}],\label{gFunc}
\end{align}
where $\mb{X}\sim P_{\mb{X}}$,  $P_{\mb{X}} \to Q_{\mb{Y}_j|\mb{X}}\to P_{\mb{Y}_j}$ and $\bsigma_{\mb{X}|U}:=\mathbb{E}[\Cov(\mb{X}|U)]$.

\begin{defn}\label{defn_ND}
We say $(Q_{\mb{Y}_1|\mb{X}},\dots,Q_{\mb{Y}_m|\mb{X}})$ is \emph{non-degenerate} if each $Q_{\mb{Y}_j|\mb{X=0}}$ is a $n_j$-dimensional Gaussian distribution with invertible covariance matrix.
\end{defn}

We have an \emph{extremisability} result under a covariance constraint:

\begin{thm}\label{thm:GaussEntropy}
If $(Q_{\mb{Y}_1|\mb{X}},\dots,Q_{\mb{Y}_m|\mb{X}})$ are non-degenerate, then $\inf_{P_{\mb{X}U}} \{ F(P_{\mb{X}U})   \colon  \bsigma_{\mb{X}|U}\preceq \mb{\Sigma} \}$ is finite and is attained by a Gaussian $\mb{X}$ and constant $U$.
\end{thm}
\begin{proof}[Proof Sketch]
Assume that both $P_{\mb{X}^{(1)}U^{(1)}}$ and $P_{\mb{X}^{(2)}U^{(2)}}$ are minimizers of \eqref{gFunc} subject to $ \bsigma_{\mb{X}|U}\preceq \bsigma$ (see \cite{lccv2015} for the proof of the existence of minimizer). Let
\begin{align}
(U^{(1)}, \mb{X}^{(1)}, \mb{Y}_1^{(1)}, \dots, \mb{X}_m^{(1)})
&\sim P_{\mb{X}^{(1)}U^{(1)}}Q_{\mb{Y}_1|\mb{X}}\dots Q_{\mb{Y}_m|\mb{X}}
\nonumber
\\
(U^{(2)}, \mb{X}^{(2)}, \mb{Y}_1^{(2)}, \dots, \mb{X}_m^{(2)})
&\sim P_{\mb{X}^{(2)}U^{(2)}}Q_{\mb{Y}_1|\mb{X}}\dots Q_{\mb{Y}_m|\mb{X}}
\nonumber
\end{align}
be mutually independent. Define
$
{ \mb X}^\pm = \frac{1}{\sqrt{2}}\left( \mb{X}^{(1)} \pm   \mb{X}^{(2)}\right)
$.
Define $\mb{Y}_j^{+}$ and $\mb{Y}_j^{-}$ similarly for $j=1, \dots, m$, and put $\hat{U}=(U^{(1)},U^{(2)})$.  We now observe that
\begin{enumerate}
\item Due to the Gaussian nature of $Q_{\mb{Y}_j|\mb{X}}$, $\mb{Y}_j^{+}|\{\mb{X}^+=\mb{x}^+,\mb{X}^-=\mb{x}^-,
    \hat{U}=u\}\sim Q_{\mb{Y}_j|\mb{X}=\mb{x}^+}$ is independent of $\mb{x}^-$.
    Thus $\mb{Y}_j^{+}|\{\mb{X}^+=\mb{x},
    \hat{U}=u\}\sim Q_{\mb{Y}_j|\mb{X}=\mb{x}}$ as well.
Similarly, $\mb{Y}_j^{-}|\{\mb{X}^-=\mb{x},\hat{U}=u\}\sim Q_{\mb{Y}_j|\mb{X}=\mb{x}}$.
\item $\bsigma_{\mb{X}^+|\uh},\,\bsigma_{\mb{X}^{-}|\mb{X}^+\hat{U}}\preceq  \bsigma_{\mb{X}^{-}|\hat{U}}$ so both $P_{\mb{X}^+,\hat{U}}$ and $P_{\mb{X}^-,\hat{U}\mb{X}^+}$ satisfy the covariance constraint.
\end{enumerate}
Using steps similar to the conventional converse proofs in information theory, one can show that
\begin{align}
\sum_{i=1}^2F(P_{\mb{X}^{(i)}U^{(i)}})
\geq F(P_{\mb{X}^+,\hat{U}})+F(P_{\mb{X}^-,\hat{U}\mb{X}^+}).
\label{e_chain}
\end{align}
But both $P_{\mb{X}^+,\uh}$ and $P_{\mb{X}^-,\uh \mb{X}^+}$ are candidate optimizers of \eqref{gFunc} subject to the given covariance constraint
whereas $P_{U^{(i)}\mb{X}^{(i)}}$ are the optimizers by assumption ($i=1,2$), so
\begin{align}
\max_i F(P_{\mb{X}^{(i)}U^{(i)}})\le \min\{F(P_{\mb{X}^+,\hat{U}}),F(P_{\mb{X}^-,\hat{U}\mb{X}^+})\},
\end{align}
which combined with \eqref{e_chain} implies that $F(\cdot)$ has the same value at $P_{\mb{X}^{(1)}U^{(1)}}$, $P_{\mb{X}^{(2)}U^{(2)}}$, $P_{\mb{X}^+,\hat{U}}$ and $P_{\mb{X}^-,\hat{U}\mb{X}^+}$. We now need the following basic observation to conclude that \emph{each term} in the linear combination in the definition of $F(\cdot)$ is also equal under those four distributions.
\begin{lem}\label{lem:invariance}
Let $p$ and $q$ be real-valued functions on an arbitrary set $\mathcal{D}$.  If
$f(t):= \min_{x\in \mathcal{D}} \{p(x) + t q(x)\}$
is attained for all $t\in(0,\infty)$,
then for almost all $t$, $f'(t)$ exists and equals
\begin{align}
q(x^{\star}) ~~~\forall x^{\star}\in \arg \min_{x\in \mathcal{D}} \{p(x) + t q(x)\}.
\end{align}
In particular, for all such $t$, $
q(x^{\star}) = q(\tilde{x}^{\star})$ and
$p(x^{\star}) = p(\tilde{x}^{\star})$ for all $x^{\star},\tilde{x}^{\star}\in \arg \min_{x\in \mathcal{D}} \{p(x) + t q(x)\}$.
\end{lem}
Geometrically $f(\cdot)$ can be viewed as the negative of the Legendre-Fenchel transformation\footnote{An extension of the Legendre-Fenchel transformation of a convex set.} of $\mathcal{S}:=\{(-q(x),p(x))\}_{x\in\mathcal{D}}$. Hence $f(\cdot)$ is concave, and the left and the right derivatives are determined by the two extreme points of the intersection between $\mathcal{S}$ and the supporting hyperplane.
See \cite{lccv2015} for a complete proof of Lemma~\ref{lem:invariance}.

By Lemma \ref{lem:invariance} and symmetry,
for almost all $(c_0,\dots,c_m)$,
\begin{align}
&h(\mb{X}^{+}|\uh) = h(\mb{X}^{-}|\mb{X}^{+},\uh) = h(\mb{X}^{+}|\mb{X}^{-},\uh)
\nonumber
\\
\Longrightarrow &I(\mb{X}^+;\mb{X}^-|\uh)=0.
\end{align}
The proof is completed by a strengthening of the Skitovic-Darmois characterization of normal distributions \cite{geng2014yanlin}:
\begin{lem}\label{SDcharacterization}
If $\mathbf{A}_1$ and $\mathbf{A}_2$ are mutually independent random vectors such that $\mathbf{A}_1+\mathbf{A}_2$ is independent of $\mathbf{A}_1-\mathbf{A}_2$, then  $\mathbf{A}_1$ and $\mathbf{A}_2$ are normally distributed with identical covariances.
\end{lem}
\end{proof}

\begin{rem}\label{rem_added}
New ingredients added to the Geng-Nair approach \cite{geng2014yanlin}\cite{nairextremal} for establishing Gaussian optimality include:
\begin{itemize}
\item Lemma~\ref{lem:invariance}, that is, by differentiating with respect to the linear coefficients, we can conveniently obtain information theoretic identities which helps us to conclude the conditional independence of $\mb{X}^+$ and $\mb{X}^-$ quickly.\footnote{The idea of differentiating the coefficients has been used in \cite{courtade2014extremal}.}
For small $m$, in principle, this may be avoided
by exhaustively enumerating the expansions of two-letter quantities (e.g.~as done in \cite{nairextremal}), but that approach becomes increasingly complicated and unstructured as $m$ increases.
\item A semicontinuity property is used in the proof of the existence of the minimizer (not discussed here, but see \cite{lccv2015}).
The continuity of differential entropy argument in \cite{geng2014yanlin}\cite{nairextremal} does not apply here for a sequence of weakly convergent $\mb{X}_n$, since their densities are not regularized by convolving with the Gaussian density.
\end{itemize}
\end{rem}

If we do not impose the non-degenerate assumption and the regularization $\bsigma_{\mb{X}|U}\preceq \mb{\Sigma}$, it is possible that the optimization in Theorem~\ref{thm:GaussEntropy} is nonfinite and/or not attained by any $P_{U\mb{X}}$. In this case, we can prove that the optimization is \emph{exhausted} by Gaussian distributions, by taking the limit in Theorem~\ref{thm:GaussEntropy} as the variance of the additive noise converges to zero.
To state the result conveniently, for any $P_{\mb{X}}$, define
\begin{align}
F_0(P_{\mb{X}}):=-h(\mb{X})+\sum_{j=1}^m c_j h(\mb{Y}_j)+c_0\Tr[\mb{M}\bsigma_{\mb{X}}],
\label{e_103}
\end{align}
where $(\mb{X},\mb{Y}_j)\sim P_{\mb{X}}Q_{\mb{Y}_j|\mb{X}}$.

\begin{thm}\label{thm:exhaustibility}
In the general (possibly degenerate) case,
For any given positive semidefinite $\mb{\Sigma}$,
\begin{align}
\inf_{P_{\mb{X}U}, \bsigma_{\mb{X}|U}\preceq \bsigma} F(P_{\mb{X}U})
=
\inf_{P_{\mb{X}}\textrm{ Gaussian}, \bsigma_{\mb{X}|U}\preceq \bsigma} F_0(P_{\mb{X}}).
\label{e_78}
\end{align}
The same holds when the covariance constraint is dropped.
\end{thm}
Note that Theorem~\ref{thm:exhaustibility} reduces an infinite dimensional optimization problem to a finite dimensional one.
From
Theorem~\ref{thm:GaussEntropy}
and
Theorem~\ref{thm:exhaustibility}
one easily obtains  Gaussian optimality results in a related optimization problem involving mutual information; see \cite{lccv2015} for details.
We close the section by mentioning several implications of the Gaussian optimality results:
\begin{itemize}
\item Extension of BL to Gaussian transformations: when $Q_{\bf X}$ and $(Q_{{\bf Y}_j|\bf X})$ are Gaussian, \eqref{e_func} holds if and only if it holds for all Gaussian functions $(f_j)$.
\item Multivariate Gaussian hypercontractivity: we say an $m$-tuple of random variables $(X_1,\dots,X_m)\sim Q_{X^m}$ is $(p_1,\dots,p_m)$-hypercontractive for $p_j\in[1,\infty]$ if
\begin{align}\label{hyper1}
\mathbb{E}\left[\prod_{j=1}^m f_j(X_j)\right]\le \prod_{j=1}^m\|f_j(X_j)\|_{p_j}
\end{align}
for all bounded measurable $(f_j)$. Suppose $Q_{X^m}=\mathcal{N}(\mb{0},\mb{\Sigma})$ where $\bf \Sigma$ is a positive semidefinite matrix whose diagonal values are all $1$.
By choosing an appropriate $\bf M$ in \eqref{e_103}, one sees that Theorem~\ref{thm:exhaustibility} continues to hold if the differential entropies are replaced by relative entropies with Gaussian reference measures.
Then by the dual formulation of hypercontractivity,
$(p_j)$ is hypercontractive if
and only if a certain matrix inequality is satisfied, which can be simplified \cite{lccv2015} to the following condition:
\begin{align}\label{e87}
\mb{P}\succcurlyeq\mb{\Sigma},
\end{align}
that is, $\mb{P}-\mb{\Sigma}$ is positive semidefinite,
where $\mb{P}$ is a diagonal matrix with $(p_j)$ on its diagonal. The $m=2$ case is  Nelson's hypercontractivity theorem.
\item The rate region for certain network information theory problems can be solved by finite dimensional matrix optimizations. For example, the multivariate Wyner common information \cite{xu2013wyner} of $m$ Gaussian scalar random variables $X_1,\dots,X_m$ with covariance matrix $\bsigma\succ\mb{0}$ is given by
    \begin{align}
    \frac{1}{2}\inf_{\bf\Lambda}\log\frac{|\bsigma|}{|\bf\Lambda|}
    \label{e137}
    \end{align}
    where the infimum is over all diagonal matrices $\bf \Lambda$ satisfying $\bf \Lambda\preceq\Sigma$. Previously, an estimation theoretic argument \cite[Corollary~1]{xu2013wyner} only establishes the Gaussian optimality of the auxiliary provided that $\bsigma$ satisfies certain conditions.
    Other examples include
    one communicator common randomness generation
    and omniscient helper key generation \cite{liu2015key}.
\end{itemize}

\section{Duality Result for the Reverse Inequality}\label{sec_rev}
In this section we give a dual to Theorem~\ref{thm_1}.
Here we state and sketch proof for finite $\mathcal{X}^m$ - an assumption used in showing a ``splitting'' property of relative information.
Extension to more general alphabets is treated in \cite{lccv2015}.
\begin{thm}\label{thm_rv}
Fix $Q_{Y|X^m}$, $(Q_{X_j})$, $R_Y$ and $d\in\mathbb{R}$. Assume $|\mathcal{X}^m|<\infty$ and $Q_Y\ll\mu$ where $\prod_j Q_{X_j}\to Q_{Y|X^m}\to Q_Y$. The following two statements are equivalent:
\begin{enumerate}
  \item For any $f_j\colon \mathcal{X}_j\to [0,+\infty)$, if $F\colon \mathcal{Y}\to[0,+\infty)$ is such that
      $
      \mathbb{E}[\log F(Y)|X^m=x^m]
      \ge \sum_j c_j\log f_j(x_j)
      $,
      $\prod_j Q_{X_j}$-almost surely, then
      \begin{align}
      \int F{\rm d}R_Y
      \ge \exp(d)\prod_j\left(\int f_j{\rm d}Q_{X_j}\right)^{c_j}.
      \label{e_func_rv}
      \end{align}
  \item For any $(P_{X_j})$, there exists a coupling $P_{X^m}$ 
      such that
      \begin{align}
      D(P_Y\|R_Y)+d\le \sum_j c_j D(P_{X_j}\|Q_{X_j})\label{e_3}
      \end{align}
      where $P_{X^m}\to Q_{Y|X^m}\to P_Y$.
\end{enumerate}
\end{thm}
Notice that compared to Theorem~\ref{thm_1}, the inequality signs in \eqref{e_func_rv} and \eqref{e_3} are reversed, and the computation of the best constant $d$ involves an extra optimization (over $F$ or $P_{X^m}$).
\begin{proof}[Proof Sketch]
We consider $d=0$ only as the general case can be handled similarly.
1)$\Rightarrow$2)
  This is the nontrivial direction which uses the finiteness of $|\mathcal{X}^m|$.
  Given $(P_{X_j})$, suppose $P_{X^m}$ is a coupling that
  minimizes $D(P_Y\|R_Y)$ (which exists because $D(P_Y\|R_Y)$ is lower semicontinuous in $P_{X^m}$).
  It is a standard exercise using KKT conditions to show an important splitting property:
  there exist $g_j\colon\mathcal{X}_j\to\mathbb{R}$ such that
  \begin{align}
  \mathbb{E}[\imath_{P_Y\|R_Y}(Y)|X^m=x^m]\ge\sum_j c_j g_j(x_j)
  \label{e_split}
  \end{align}
  for all $x^m$, and the equality holds $P_{X^m}$-almost surely.\footnote{This property is related (but not equivalent) to a property of $I$-projections onto linear sets discussed by Csisz\'{a}r \cite[Section~3]{csiszar1975divergence}; see \cite{lccv2015} for a discussion.}
  The latter claim follows by applying complementary slackness to the nonnegativity constraint on $P_{X^m}$.
  Now define $S_{X_j}$ by
  \begin{align}
  \imath_{S_{X_j}\|Q_{X_j}}(x_j)=-d_j+g_j(x_j)
  \end{align}
  where $d_j$'s are normalization constants.
  Applying 1) with $F\leftarrow \frac{{\rm d}P_Y}{{\rm d}R_Y}$, $f_j\leftarrow \exp(g_j)$,
  and using the fact that $D(P_{X_j}\|S_{X_j})\ge 0$, we obtain 2) upon rearranging.

2)$\Rightarrow$1) Given $F$, $(f_j)$, define $S_Y$, $(P_{X_j})$ by
      \begin{align}
      \imath_{S_Y\|R_Y}(y)&=-d_0+\log F(y);
      \\
      \imath_{P_{X_j}\|Q_{X_j}}(x_j)&=-d_j+\log f_j(x_j),
      \end{align}
      where $d_0$,\dots,$d_m$ are normalization constants.
      The assumption $Q_Y\ll\mu$ guarantees that $S_Y$ and $(P_{X_j})$ are well-defined except for the trivial case where $F$ and some $f_j$ are zero almost surely.
      Then choose the $P_Y$ such that \eqref{e_3} holds.
     Finally 1) follows from 2) and the nonnegativity of $D(P_Y\|S_Y)$.
\end{proof}
\begin{rem}
Once the finiteness assumption on $\mathcal{X}^m$ is dropped (see \cite{lccv2015}), we can recover Barthe's formulation of the reverse Brascamp-Lieb inequality \cite{barthe1998reverse} from \eqref{e_func_rv}: when $Q_{Y|X^m}$ is deterministic given by $\phi\colon \mathcal{X}^m\to\mathcal{Y}$, then the first statement in Theorem~\ref{thm_rv} holds if and only if it holds for
\begin{align}
F(y)
:= \sup_{x^m\colon\phi(x^m)=y}\prod_j f_j^{c_j}(x_j),
\quad\forall y.
\label{e_rv_f}
\end{align}
The RBL is recovered by letting $\phi(\cdot)$ be a linear function.
\end{rem}
\begin{rem}
Both the forward and reverse duality theorems recover the strong data processing inequality when $m=1$ (that is, the MAC or the broadcast channel is a point-to-point channel). No meaningful version of reverse strong data processing is immediately apparent; the naive candidate which simply reverses the inequality sign doesn't even tensorize \cite{Liu}.
\end{rem}

\bibliographystyle{ieeetr}
\bibliography{ref_om}

\begin{thebibliography}{10}

\bibitem{brascamp1976best}
H.~J. Brascamp and E.~H. Lieb, ``Best constants in {Young's} inequality, its
  converse, and its generalization to more than three functions,'' {\em
  Advances in Mathematics}, vol.~20, no.~2, pp.~151--173, 1976.

\bibitem{lieb1990gaussian}
E.~H. Lieb, ``Gaussian kernels have only {Gaussian} maximizers,'' {\em
  {Inventiones Mathematicae}}, vol.~102, no.~1, pp.~179--208, 1990.

\bibitem{bennett2008brascamp}
J.~Bennett, A.~Carbery, M.~Christ, and T.~Tao, ``The {Brascamp--Lieb}
  inequalities: finiteness, structure and extremals,'' {\em Geometric and
  Functional Analysis}, vol.~17, no.~5, pp.~1343--1415, 2008.

\bibitem{ISIT_lccv_smooth2016}
J.~Liu, T.~A. Courtade, P.~Cuff, and S.~Verd\'{u}, ``{Smoothing Brascamp-Lieb
  inequalities and strong converses for CR generation},'' in {\em Proc.~of IEEE
  International Symposium on Information Theory}, July 2016, Barcelona, Spain.

\bibitem{courtade2013outer}
T.~Courtade, ``Outer bounds for multiterminal source coding via a strong data
  processing inequality,'' in {\em Proc.~of IEEE International Symposium on
  Information Theory}, pp.~559--563, July 2013, Istanbul, Turkey.

\bibitem{pw_2015}
Y.~Polyanskiy and Y.~Wu, ``A note on the strong data-processing inequalities in
  {Bayesian} networks,'' {\em http://arxiv.org/pdf/1508.06025v1.pdf}.

\bibitem{liu2015key}
J.~Liu, P.~Cuff, and S.~Verd\'{u}, ``Secret key generation with one
  communicator and a one-shot converse via hypercontractivity,'' in {\em
  Proc.~of 2015 IEEE International Symposium on Information Theory},
  pp.~710--714, June 2015, Hong Kong, China.

\bibitem{xu15}
A.~Xu and M.~Raginsky, ``Converses for distributed estimation via strong data
  processing inequalities,'' in {\em Proc.~of the 2015 IEEE International
  Symposium on Information Theory (ISIT), Hong Kong, China}, July 2015.

\bibitem{ganor2014exponential}
A.~Ganor, G.~Kol, and R.~Raz, ``{Exponential separation of information and
  communication},'' in {\em Foundations of Computer Science (FOCS), 2014 IEEE
  55th Annual Symposium on}, pp.~176--185, 2014.

\bibitem{braverman2015communication}
M.~Braverman, A.~Garg, T.~Ma, H.~L. Nguyen, and D.~P. Woodruff, ``Communication
  lower bounds for statistical estimation problems via a distributed data
  processing inequality,'' {\em arXiv preprint arXiv:1506.07216}, 2015.

\bibitem{mossel2010noise}
E.~Mossel, R.~O'Donnell, and K.~Oleszkiewicz, ``Noise stability of functions
  with low influences: Invariance and optimality,'' {\em Annals of
  Mathematics}, vol.~171, no.~1, pp.~295--341, 2010.

\bibitem{duchi13}
M.~J. John C~Duchi and M.~J. Wainwright, ``Local privacy and statistical
  minimax rates,'' in {\em IEEE 54th Annual Symposium on Foundations of
  Computer Science (FOCS)}, pp.~429--438, 2013.

\bibitem{geng2014yanlin}
Y.~Geng and C.~Nair, ``The capacity region of the two-receiver {Gaussian}
  vector broadcast channel with private and common messages,'' {\em IEEE
  Transactions on Information Theory}, vol.~60, no.~4, pp.~2087--2104, April,
  2014.

\bibitem{nairextremal}
C.~Nair, ``An extremal inequality related to hypercontractivity of {Gaussian}
  random variables,''

\bibitem{barthe1998optimal}
F.~Barthe, ``Optimal {Young's} inequality and its converse: a simple proof,''
  {\em {Geometric and Functional Analysis}}, vol.~8, no.~2, pp.~234--242, 1998.

\bibitem{barthe1998reverse}
F.~Barthe, ``On a reverse form of the {Brascamp-Lieb} inequality,'' {\em
  {Inventiones Mathematicae}}, vol.~134, no.~2, pp.~335--361, (see also
  arXiv:math/9705210 [math.FA]), 1998.

\bibitem{kamath_reverse}
S.~Kamath, ``Reverse hypercontractivity using information measures,'' in {\em
  Proc.~of the 53rd Annual Allerton Conference on Communications, Control and
  Computing}, pp.~627--633, Sept.~30-Oct.~2, 2015, UIUC, Illinois.

\bibitem{BN2015}
C.~Nair, ``Tensorization: information theory and hypercontractivity.''
  \url{http://www.ee.iitb.ac.in/bits/wp-content/uploads/2016/01/BITS.pdf}.
\newblock based on work by S.~Beigi and C.~Nair in 2015, presented at Bombay
  Information Theory Seminar 2016.

\bibitem{lccv2015}
J.~Liu, T.~A. Courtade, P.~Cuff, and S.~Verd\'{u}, ``Information theoretic
  perspectives on {Brascamp-Lieb} inequality and its reverse,'' {\em draft}.

\bibitem{carlen2009subadditivity}
E.~A. Carlen and D.~Cordero-Erausquin, ``Subadditivity of the entropy and its
  relation to {Brascamp--Lieb} type inequalities,'' {\em Geometric and
  Functional Analysis}, vol.~19, no.~2, pp.~373--405, 2009.

\bibitem{ahlswede1976spreading}
R.~Ahlswede and P.~G\'{a}cs, ``Spreading of sets in product spaces and
  hypercontraction of the {Markov} operator,'' {\em The Annals of Probability},
  vol.~4, pp.~925--939, 1976.

\bibitem{nair}
C.~Nair, ``Equivalent formulations of hypercontractivity using information
  measures,'' {\em International Zurich Seminar}, Feb. 2014.

\bibitem{atarRobust}
R.~Atar, K.~Chowdhary, and P.~Dupuis, ``Robust bounds on {risk-sensitive}
  functionals via {R\'{e}nyi} divergence,'' {\em SIAM J. Uncertainty Quant.},
  vol.~3, pp.~18--33, 2015.

\bibitem{atar2014information}
R.~Atar and N.~Merhav, ``Information-theoretic applications of the logarithmic
  probability comparison bound,'' {\em IEEE Transactions on Information
  Theory}, vol.~61, no. 10, pp.~5366--5386, Oct. 2015.

\bibitem{csiszar2011information}
I.~Csiszar and J.~K{\"o}rner, {\em Information theory: coding theorems for
  discrete memoryless systems (Second edition)}.
\newblock Cambridge University Press, 2011.

\bibitem{anan_13}
V.~Anantharam, A.~Gohari, S.~Kamath, and C.~Nair, ``On maximal correlation,
  hypercontractivity, and the data processing inequality studied by {Erkip and
  Cover},'' {\em http://arxiv.org/pdf/1304.6133v1.pdf}.

\bibitem{loomis1949}
L.~H. Loomis and H.~Whitney, ``An inequality related to the isoperimetric
  inequality,'' {\em Bull. Amer. Math. Soc.}, vol.~55, pp.~961--962, 1949.

\bibitem{shearer}
P.~F. F.~Chung, R.~Graham and J.~Shearer, ``Some intersection theorems for
  ordered sets and graphs,'' {\em J. Combinatorial Theory Series A}, vol.~43,
  no.~1, pp.~23--37, 1986.

\bibitem{courtade2014extremal}
T.~Courtade and J.~Jiao, ``An extremal inequality for long {M}arkov chains,''
  in {\em Proc.~of the 52rd Annual Allerton Conference on Communications,
  Control and Computing}, pp.~763--770, Oct.~1-3, 2014, UIUC, Illinois.

\bibitem{xu2013wyner}
G.~Xu, W.~Liu, and B.~Chen, ``Wyner's common information: Generalizations and a
  new lossy source coding interpretation,'' {\em arXiv preprint
  arXiv:1301.2237}, 2013.

\bibitem{csiszar1975divergence}
I.~Csisz{\'a}r, ``I-divergence geometry of probability distributions and
  minimization problems,'' {\em The Annals of Probability}, pp.~146--158, 1975.

\bibitem{Liu}
J.~Liu, P.~Cuff, and S.~Verd\'{u}, ``Key capacity for product sources with
  application to stationary {G}aussian processes,'' {\em IEEE Transactions on
  Information Theory}, vol.~62, pp.~984--1005, Feb.~2016.

\end{thebibliography}
\end{document}